\documentclass[11pt]{article}
\usepackage{amsfonts,amssymb,a4,bm}
\usepackage[all]{xy}
\usepackage{amsmath}
\usepackage{color}

\newcommand{\qed}{\hfill\rule{3mm}{3mm}}
\newtheorem{teorema}{Theorem}
\newtheorem{defi}{Definition}
\newtheorem{cor}{Corollary}
\newtheorem{lem}{Lemma}

\newtheorem{pro}{Proposition}
\makeatletter \@addtoreset{equation}{section} \makeatother
\begin{document}


\voffset=-1.5truecm\hsize=16.5truecm    \vsize=24.truecm
\baselineskip=14pt plus0.1pt minus0.1pt \parindent=12pt
\lineskip=4pt\lineskiplimit=0.1pt      \parskip=0.1pt plus1pt

\def\ds{\displaystyle}\def\st{\scriptstyle}\def\sst{\scriptscriptstyle}


\let\a=\alpha \let\b=\beta \let\ch=\chi \let\d=\delta \let\e=\varepsilon
\let\f=\varphi \let\g=\gamma \let\h=\eta    \let\k=\kappa \let\l=\lambda
\let\m=\mu \let\n=\nu \let\o=\omega    \let\p=\pi \let\ph=\varphi
\let\r=\rho \let\s=\sigma \let\t=\tau \let\th=\vartheta
\let\y=\upsilon \let\x=\xi \let\z=\zeta
\let\D=\Delta \let\F=\Phi \let\G=\Gamma \let\L=\Lmbda \let\Th=\Theta
\let\O=\Omega \let\P=\Pi \let\Ps=\Psi \let\Si=\Sigma \let\X=\Xi
\let\Y=\Upsilon\let\L\Lambda



\def\\{\noindent}
\let\io=\infty

\def\VU{{\mathbb{V}}}
\def\ED{{\mathbb{E}}}
\def\GI{{\mathbb{G}}}
\def\Tt{{\mathbb{T}}}
\def\C{\mathbb{C}}
\def\LL{{\cal L}}
\def\RR{{\cal R}}
\def\SS{{\cal S}}
\def\NN{{\cal M}}
\def\MM{{\cal M}}
\def\HH{{\cal H}}
\def\GG{{\cal G}}
\def\PP{{\cal P}}
\def\AA{{\cal A}}
\def\BB{{\cal B}}
\def\FF{{\cal F}}
\def\TT{{\cal T}}
\def\v{\vskip.1cm}
\def\vv{\vskip.2cm}
\def\gt{{\tilde\g}}
\def\E{{\mathcal E} }
\def\I{{\rm I}}
\def\0{\emptyset}
\def\xx{{\V x}} \def\yy{{\bf y}} \def\kk{{\bf k}} \def\zz{{\bf z}}
\def\ba{\begin{array}}
\def\ea{\end{array}}  \def \eea {\end {eqnarray}}\def \bea {\begin {eqnarray}}
\def\xto#1{\xrightarrow{#1}}

\def\tende#1{\vtop{\ialign{##\crcr\rightarrowfill\crcr
              \noalign{\kern-1pt\nointerlineskip}
              \hskip3.pt${\scriptstyle #1}$\hskip3.pt\crcr}}}
\def\otto{{\kern-1.truept\leftarrow\kern-5.truept\to\kern-1.truept}}
\def\arm{{}}
\font\bigfnt=cmbx10 scaled\magstep1

\newcommand{\card}[1]{\left|#1\right|}
\newcommand{\und}[1]{\underline{#1}}
\def\1{\rlap{\mbox{\small\rm 1}}\kern.15em 1}
\def\ind#1{\1_{\{#1\}}}
\def\bydef{:=}
\def\defby{=:}
\def\buildd#1#2{\mathrel{\mathop{\kern 0pt#1}\limits_{#2}}}
\def\card#1{\left|#1\right|}
\def\proof{\noindent{\bf Proof. }}
\def\qed{ \square}
\def\reff#1{(\ref{#1})}
\def\eee{{\rm e}}
\def\be{\begin{equation}}
\def\ee{\end{equation}}

\title{\Large Convergence of  Mayer  and Virial expansions
and the  Penrose tree-graph identity}

\author{\normalsize Aldo Procacci\footnote{\scriptsize Departamento de Matem{\'a}tica, Universidade Federal de Minas Gerais, Belo Horizonte-MG, Brazil - aldo@mat.ufmg.br},
Sergio A. Yuhjtman\footnote{\scriptsize Departamento de Matem\'atica, Universidad de Buenos Aires, Buenos Aires, Argentina - sergioyuhjtman@gmail.com}}

\maketitle

\begin{abstract} We  establish  new lower bounds for the convergence radius of the Mayer series and the Virial series
of  a  continuous particle system interacting via a stable and
tempered pair potential. Our  bounds considerably improve those
given by Penrose and Ruelle in 1963 for the Mayer series and  by
Lebowitz and Penrose in 1964 for the Virial series. To get our
results we exploit the tree-graph  identity given by Penrose in 1967
using a new partition scheme based on  minimum spanning trees.
\end{abstract}

\vskip.3cm
{\footnotesize

}
\vskip.5cm

\section{Introduction} \let\thefootnote\relax\footnotetext{2010 {\it Mathematics Subject Classification.} Primary 82B21; Secondary 05C05.}
Continuous particle systems are  an important subject of
investigation in rigorous statistical mechanics.
Since the thirties it was clear that the equation of state
of a non ideal gas in thermodynamics
can be deduced from statistical mechanics via the Mayer series
of the pressure in the grand canonical ensemble.
The Mayer series, proposed by Mayer (see, e.g.,  \cite{MM} and references therein) is  an expansion in powers of the fugacity $\l$ of the
logarithm of the Grand Canonical partition function 
of the the system under study,
whose $n^{\rm th}$-order coefficient is  given by a sum of terms indexed by the labeled connected graphs between $n$ vertices.
This series was at the time only formal and the question about  its  convergence
uniformly in the volume where the gas is confined, remained an enigma during two decades.
The difficulty  in dealing with this problem was basically due to challenging  combinatorial structure
of the Mayer coefficients.

\\In the sixties, Penrose \cite{Pe63} and independently Ruelle \cite{Ru63}  proved the convergence of the Mayer  series
for a very large class of continuous particle systems. Namely those
interacting via a stable and tempered pair potential. The bound on
the  convergence radius $R$ of the Mayer series  obtained in
\cite{Pe63} and \cite{Ru63} still stands as the best available in
the literature. It is worth to remind that the first singularity of
Mayer series is in general not on the positive real fugacity axis
(e.g. it is in the negative axis at $\l=-R$ for repulsive
potentials, see \cite{Ru}) so that the value $\l=R$ for the fugacity
is not directly related to any physical phase transition point.
However a  lower bound for $R$ as sharp as possible is  a physically
relevant information since it allows to maximize  the  region of
analyticity of the pressure of the system, i.e. where the system is
surely in the gas phase and no phase transitions occur. In
this respect, there have been only very few improvements on the
Penrose-Ruelle bound for $R$ by considering some restricted
sub-classes of stable and  tempered potentials. Basically such
improvements can be summarized as follows. Brydges and collaborators
improved the Penrose-Ruelle bound for absolutely summable pair
potential (see e.g.  \cite{BM} and references therein) and Basuev
gave in \cite{Ba2} an improvement for a significant class of
potentials, introduced by himself in \cite{Ba1}, which is
sufficiently large to include Lennard-Jones type potentials  (see
also recently \cite{dLPY}).

\\The method used to get the general  bounds of 1963 were based on the so-called Kirkwood-Salzburg equations
(see \cite{Ru} and references therein). It is an indirect method which aims
to control simultaneously all correlation functions of the systems by analyzing an infinite set of equations involving them.
On the other hand,
the methods used to get improvements cited above for subclasses of  stable and tempered potentials are all based
on trying  to bound directly the Mayer coefficients.
The main technical tools  to obtain this task are the so-called tree-graph identities. Several
alternative identities have been proposed along the last four decades   (see  \cite{BM,dLPY} and references therein).

\\The very first of these tree-graph identities, proposed by Penrose \cite{Pe67} in 1967,
was based on the existence of a map $\bm M$ (a so-called {\it partition scheme})
from the set ${\mathcal T}_n$ of the labeled trees with $n$ vertices to the set ${\mathcal G}_n$
of the labeled connected graphs  with $n$ vertices. This map is able to  induce a partition of the set ${\mathcal G}_n$ into blocks
indexed by the elements of set ${\mathcal T}_n$ (details in Section \ref{sec3}).
Penrose gave an explicit construction of such a map and used it
to rewrite the Mayer coefficients as sum of terms indexed by trees rather than by connected graphs.
Using this alternative expression of the Mayer coefficients, Penrose was then able to bound them directly,
reobtaining the same bound of 1963 but at the cost to   impose a further condition on the pair potential, beyond stability and temperedness.
Namely, the pair potential needed to have an  hard-core.
This fact, we think, has led researchers to believe that the identity proposed by Penrose was  useful only for systems
with hard-core and in fact in the literature it is mentioned only for such kind of systems (see, e.g.,  \cite{FP,So1} and reference therein).

\\In the present paper we disclaim this belief  and show that
the Penrose tree-graph identity  can be used
to strongly improve the old Penrose-Ruelle bounds
for general stable and tempered potentials.
The main idea is based on the
definition of a new map $\bm M$, different from the one originally proposed by Penrose,
which still is able to make a partition of the set of connected graphs whose blocks are  indexed
by  minimum spanning trees, where the minimality is basically on the energies of the edges of the tree.
Such new partition scheme, as we will see,  allows to use in an optimal
way the stability condition of the pair potential.

\\The rest of the paper is organized as follows. In Section \ref{sec2} we remind the basic concepts
for  continuous systems of classical particles, we give the definitions of stability
and temperedness and write down the Mayer series explicitly. We end the section by recalling the Penrose-Ruelle result (Theorem
0) and by stating the main results of the paper. Namely, in Theorem \ref{teo1} we present the new bound on the
$n^{\rm th}$-order Mayer coefficient and consequently the new lower bound on the convergence
radius of the Mayer series
 and in Proposition \ref{pro1} we give the main technical tool to proof Theorem \ref{teo1}, {\it id est} a new tree graph inequality, deduced from the
Penrose tree-graph identity.
In Section 3 we prove Proposition  \ref{pro1}. We start by recalling  the Penrose identity
for a generic partition scheme $\bm M$, we then introduce our new partition scheme
and show how to use it to obtain the inequality given in Proposition \ref{pro1}. In Section 4 we conclude the
proof of Theorem \ref{teo1}. Finally, in Section 5 we give a comparison with previous results
and add some concluding remarks.

\section{Model and Results}\label{sec2}

\subsection{Model}
We  consider a continuous system of  classical particles in the $d$-dimensional Euclidean space $ \mathbb{R}^d$. Denote by $x_i\in \mathbb{R}^d$
the position vector of the $i^{\rm th}$ particle of the system and by $\|x_i\|$ its Euclidean norm. We  suppose
hereafter that particles interact through a pair potential $V:\mathbb{R}^d\to \mathbb{R}\cup\{+\infty\}$
such that $V(-x)=V(x)$ and we set shortly $\mathbb{R}^*\doteq \mathbb{R}\cup\{+\infty\}$. Given
a configuration
$(x_1,\dots,x_n)\in {\mathbb R}^{dn}$ of the system such that  $n$ particles are present, the energy $U(x_1,\dots,x_n)$
of this configuration is defined as
$$
U(x_1,\dots,x_n)=\sum_{1\le i< j\le n}V(x_i-x_j)
$$
We work in the Grand Canonical Ensemble, where the statistical mechanics of the system  is  governed by the following partition function
\be\label{1.1}
\Xi_{\L}(\l,\b)=\sum_{n=0}^{\infty}\frac{\l^{n}}{n!} \int_\L dx_1
\dots \int_\L dx_n e^{-\b U(x_1,\dots,x_n)}
\ee
where $\L\subset{\mathbb R}^d$ is typically    a $d$-dimensional cube with center in the origin,
 $\lambda\in (0,+\infty)$ is the activity of the system and $\b\in (0,+\infty)$ is the inverse temperature.

\\The connection with thermodynamics is obtained by taking the logarithm of the partition function.
Namely, the pressure  $P_\L(\b,\l)$ and the density $\r_\L(\b,\l)$ of the system at fixed values of the
thermodynamic parameters inverse temperature $\b$, fugacity  $\l$ and volume $\L$,
are given  by
\be\label{press}
P_\L(\b,\l)= {1\over \b |\L|}\log \Xi_{\L}(\l,\b)
\ee
\be\label{dens}
\r_\L(\b,\l)~=~ {\l\over|\L|}{\partial\over \partial\l}\log \Xi_{\L}(\l,\b)
\ee

\\It is a long known fact  \cite{MM} that $\log \Xi_{\L}(\l,\b)$ can be expanded as a power series of the fugacity. Namely,
\be\label{pressm}
{1\over|\L|}\log \Xi_{\L}(\l,\b) ~=~ \l+ \sum_{n=2}^{\infty}C_n(\b,\L)\l^n
\ee
The series above is the so-called Mayer series and its coefficients $C_n(\b,\L)$, known nowadays as {\it Mayer
coefficients}, admit the following explicit expression.
\be\label{urse}
C_n(\b,\L)~=~{1\over |\L|}{1\over n!}\int_{\L}\,dx_1
\dots \int_{\L} dx_n\:\sum\limits_{g\in {\mathcal G}_{n}}
\prod\limits_{\{i,j\}\in E_g}\left[  e^{ -\b V(x_i-x_j)} -1\right]
\ee
where ${\mathcal G}_n$ denotes the set of all connected graphs with vertex set
$\{1,2,\dots,n\}$ (a.k.a. labeled connected graphs with $n$ vertices) and $E_g$ denotes the edge set  of $g\in {\mathcal G}_n$.

\\Clearly, via Mayer series one can write both pressure (\ref{press})
and density (\ref{dens})
as  expansions in power of the fugacity $\l$.

\\Some conditions  on the potential $V$ must be imposed to have hope to control the series. Stability and temperedness are
commonly  considered as minimal   conditions to  guarantee a good
statistical mechanics behavior of the  system (see, e.g., \cite{Ru}
and \cite{Ga}).

\begin{defi}[Stability]
A pair potential $V$ is said to be stable if
\be\label{stabi}
B:=\sup_{n\ge 2}\sup_{(x_1,\dots,x_n)\in \mathbb{R}^{dn}}-{1\over n}U(x_1,\dots,x_n) < +\infty
\ee
We call $B$ the stability constant of the pair potential $V$.
\end{defi}

\begin{defi}[Temperedness]
A pair potential $V$ is said to be tempered  if there exists $r_0\ge 0$ such that
\be\label{temp}
\int_{\|x\|\ge r_0}  |V(x)| \;dx < \infty
\ee
\end{defi}
Observe that the stability constant of any tempered potential $V$ is non-negative.

\vv\vv
\\As mentioned in the introduction,
the best rigorous upper bound on $|C_n(\b,\L)|$ so far (and hence the best lower bound on the convergence radius
of the Mayer series)  for
stable and tempered pair  potentials is that
obtained by Penrose and Ruelle in 1963 \cite{Pe63,Ru63}.

\vv\vv
\\{\bf Theorem 0 (Penrose Ruelle 1963)}
{\it
Let $V$ be a stable and tempered pair  potential with  stability constant $B$. Then
the  $n$-order Mayer  coefficient $C_n(\b,\L)$ defined in (\ref{urse})
is bounded by
\be\label{bmaru}
|C_n(\b,\L)|\le e^{2\b B (n-1)}n^{n-2} {[C(\b)]^{n-1}\over n!}
\ee
where
\be\label{cb}
C(\b)=\int_{\mathbb{R}^{d}} dx ~ |e^{-\b V(x)}-1|
\ee
Therefore the Mayer series (\ref{pressm}) converges absolutely, uniformly in $\L$,
 for any complex  $\l$ inside the disk
\be\label{radm}
|\l| <{1\over e^{2\b B+1} C(\b)}
\ee
I.e. the convergence radius $R$ of the Mayer series  admits the following lower bound
\be\label{rold}
R\ge R_{\rm PR}\,\doteq\,{1\over e^{2\b B+1} C(\b)}
\ee
}

\\Observe that the non-negative quantity $C(\b)$ defined in (\ref{cb}) is  finite if $V$ is stable and tempered (see e.g. \cite{Ru}).

\subsection{Results}
\\As said in the introduction, the bound on  the absolute value of the $n^{\rm th}$-order Mayer coefficient $C_n(\b,\L)$ appearing in Theorem 0 has been obtained ``indirectly" using the method of
Kirkwood-Salzburg equations. Our  main result
is the following Theorem \ref{teo1} below which improves strongly Theorem 0 via a direct bound on $C_n(\b,\L)$
by using the old Penrose tree-graph identity proposed in \cite{Pe67}.

\begin{teorema}\label{teo1}
 Let $V$ be a stable and tempered pair  potential with  stability constant $B$. Then
the  $n$-order Mayer  coefficient $C_n(\b,\L)$
is bounded by
\be\label{bteo1}
|C_n(\b,\L)|\le e^{\b Bn}n^{n-2} {[\hat C(\b)]^{n-1}\over n!}
\ee
where
\be\label{hcb}
\hat C(\b)=\int_{\mathbb{R}^{d}} \Big[1-e^{-\b |V(x)|} \Big]dx
\ee
Therefore the Mayer series  converges absolutely, uniformly in $\L$,
for any complex  $\l$ inside the disk
\be\label{radp}
|\l| <{1\over e^{\b B+1} \hat C(\b)}
\ee
I.e. the convergence radius $R$ of the Mayer series  admits the following lower bound
\be\label{rnew}
R\ge R^{*}\doteq {1\over e^{\b B+1} \hat C(\b)}
\ee
\end{teorema}
Observe that, according to the  definitions (\ref{cb}) and (\ref{hcb}), we have  $0\le \hat C(\b)\le C(\b)$.

\vv
\\{\bf Remark 1}.
\\The improvement given by Theorem \ref{teo1} respect to Theorem 0 is
twofold. First, the factor $e^{\b B+1}$ in (\ref{rnew})  replaces
the factor $e^{2\b B+1}$ in (\ref{rold}). Second, the factor $\hat C(\b)$ in (\ref{rnew})  replaces the factor $C(\b)$ in
(\ref{rold}) and clearly, recalling their definitions (\ref{hcb}) and (\ref{cb})  respectively, we have $\hat C(\b)\le C(\b)$ where
the equality only holds  if $V$ is non-negative  (purely repulsive). Moreover observe that while $\hat C(\b)$
grows at most linearly  in $\b$, the factor $C(\b)$ can grow exponentially with $\b$ (for $V$ with an attractive part).
So, 
the ratio
$
{R^{*}
/R_{\rm PR}
}
$
always greater than one,  is the product of two factors,
 $e^{\b B}$ and $[{\hat C(\b)/ C(\b)}]$,  both
growing exponentially fast
with $\b$ when $V$ has a negative part. Due to this exponential behavior in $\b$, our improvement on the Penrose-Ruelle bound
is rather astonishing for attractive potentials (i.e. those with strictly positive stability constant $B$) at non small $\b$ (i.e. low  temperatures).
To give an idea,   for the Lennard-Jones gas at inverse temperature $\b=1$, using the value $B_{\rm LJ}=8.61$
 for its stability constant  (see \cite{JI}), our lower bound
is at least $8.5\times 10^4$ larger than the Penrose-Ruelle lower bound, while for $\b=10$
is at least $7.26\times 10^{43}$ larger than the Penrose-Ruelle lower bound. On the other hand,
for $\b$ small the exponential $e^{\b B}$ is near to one and $\hat C(\b)$ becomes very close to $C(\b)$  so that  our improvement is  much less
sensitive at high temperatures. Finally, for purely repulsive pair potentials (i.e. those with stability constant $B=0$) our bound
coincides with the Penrose-Ruelle bound since in this case $e^{\b B}=1$ and $\hat C(\b)=C(\b)$.

 \vv
\\{\bf Remark 2}. As mentioned above,
the density of the system $\r_\L(\b,\l)$ given in (\ref{dens}) can also be written, via (\ref{pressm}),
in terms of a power series of the fugacity $\l$.
This series allows to express the fugacity $\l$ as a power series of the density $\r_\L$
which, plugged into (\ref{pressm}), gives the  so-called Virial series (see, e.g., \cite{MM}),
i.e., the pressure as a function of the density and temperature, or, in other words,  the equation of state of the system.
In 1964 Lebowitz and Penrose \cite{LP} (provided
the following lower bound on the convergence radius $\cal R$ of the Virial series of a gas of particles interacting via a stable and
tempered pair potential.
\be\label{virialold}
{\cal R}\ge {\cal R}_{\rm LP}\doteq {g(e^{2\b B})\over e^{2\b B}\hat C(\b)}
\ee
with
$$
g(u)= \max_{0<w<1} {[(1+u)e^{-w} -1]w\over u}
$$
The authors  deduced such a bound, via Lagrange inversion, from the Penrose-Ruelle
bounds  (\ref{bmaru}) of the Mayer coefficients (see also \cite{Gr} for an alternative method that deal directly with the
coefficients of the Viral series written in terms of two-connected graphs).
If one redoes the calculations performed by Lebowitz and Penrose  using the new
upper bound   (\ref{bteo1}) for the absolute value of the $n^{\rm th}$-order Mayer coefficient (with $\hat C(\b)$ in place of $C(\b)$
and $e^{\b B}$ in place of $e^{2\b B}$), then one immediately obtains
 the
following lower bound for the convergence radius ${\cal R}$ of the Viral series
\be\label{virial}
{\cal R}\ge {\cal R}^{*}\doteq  {g(e^{\b B})\over e^{\b B}\hat C(\b)}
\ee
Observing that $0.1448\le g(u)\le e^{-1}$ for $u\in [1,\infty)$ (see, e.g.,
Remark 2 below Theorem 1 in \cite{MP}), we have that the ratio ${\cal R}^{*}/{\cal R}_{\rm LP}$ between the new lower bound (\ref{virial})
for the convergence radius of the Viral series
and the old Lebowitz-Penrose bound given in (\ref{virialold}) is of the same order of magnitude of the ratio ${R^*/R_{PR}}$ calculated above.

\vv
\\The proof of Theorem \ref{teo1} relies  on a non trivial tree-graph inequality on the absolute value of the integrand of the
r.h.s. of (\ref{urse}) (the so-called
Ursell functions)  as far as  stable pair potentials are considered.
This new tree-graph inequality, an additional result of the present paper which we consider interesting {\it per se},  can be stated as follows.

\begin{pro}\label{pro1}
Let $V$ be a stable pair  potential with  stability constant $B$.  Then, for any
$n\ge 2$ and any $(x_1,\dots,x_n)\in \mathbb{R}^{dn}$, the following inequality holds
\be\label{nine}
|\sum_{g \in {\mathcal G}_n} \prod_{\{i,j\} \in E_g} (e^{-\b V(x_i-x_j)}-1)|~\le ~
e^{\b Bn}\sum_{\tau \in {\mathcal T}_n} \prod_{\{i,j\} \in E_\tau} (1-e^{-\b| V(x_i-x_j)|})
\ee
where  ${\mathcal T}_n$ denotes the set of all trees with vertex set
$[n]$ (a.k.a. labeled trees with $n$ vertices).
\end{pro}

\\The proofs of Proposition \ref{pro1} and Theorem \ref{teo1} are given in sections 3 and 4 respectively.

\section{Proof of Proposition \ref{pro1}}\label{sec3}
In this section we set
 $[n]:=\{1, ... ,n\}$ and $E_n:=\{\{i,j\} :\; i,j \in [n], \ i <j \}$. We recall also that ${\mathcal G}_n$ is the set of all connected
graphs with vertex set $[n]$ (i.e labeled connected graphs with $n$ vertices) and ${\mathcal T}_n$
is the set of all trees with vertex set $[n]$ (i.e labeled trees with $n$ vertices). Given $g\in {\mathcal G}_n$, we denote
$E_g$ its edge set. Given  $g\in {\mathcal G}_n$ and $g'\in {\mathcal G}_n$ such that $E_g\subset E_{g'}$, we say that $g$ is a subgraph of $g'$ and
write $g\subset g'$. Note that the pair $({\mathcal G}_n, \subset)$ is a partially ordered set.
Given  $g,g'\in {\mathcal G}_n$  such that $g\subset g'$, we denote
$
[g,g'] =\{g''\in {\mathcal G}_n: g\subseteq g''\subseteq g'\}
$.
In other words $[g,g']$ is an interval in $({\mathcal G}_n,\subset)$.
Given a set  $S$, we denote by $|S|$ its
cardinality.
\subsection{General Penrose identity}

Penrose identity is a rewriting of the Ursell functions
$$
\Phi^T(x_1,\dots,x_n)~=~\sum\limits_{g\in {\mathcal G}_{n}}~
\prod\limits_{\{i,j\}\in E_g}\left[  e^{ -\b V(x_i -x_j)} -1\right]
$$
based on  the existence of a map ({\it partition scheme}) from the set ${\mathcal T}_n$ of the  trees with vertex set $[n]$ to the set  ${\mathcal G}_n$
of the connected graphs with vertex set $[n]$.

\begin{defi}\label{partschem} A map $\bm{M}: {\mathcal T}_n\to {\mathcal G}_n$ is called a partition scheme in the set of the connected graphs
${\mathcal G}_n$  if, for all $\tau\in {\mathcal T}_n$,
$\tau\subset \bm{M}(\t)$ (i.e. $\t$ is subgraph of $\bm{M}(\t)$)
and $ {\mathcal G}_n=\biguplus_{\tau\in {\mathcal T}_n}[\tau,\bm M(\tau)]$
where $\biguplus$ means disjoint union and $[\tau,\bm M(\tau)]=\{g\in {\mathcal G}_n: \tau\subset g\subset \bm M(\tau)\}$
 is an interval in $ {\mathcal G}_n$ (with respect to the set-inclusion).
\end{defi}

\\Once a  partition scheme in ${\mathcal G}_n$ has been given, we have the following identity

\begin{teorema}[General Penrose identity]\label{Penid} Let
$V(x)$ be a pair potential. Let $n\ge 2$.
Let  $\bm{M}: {\mathcal T}_n\to {\mathcal G}_n$ be a partition scheme in ${\mathcal G}_n$.
Then, for any  $(x_1,\dots,x_n)\in \mathbb{R}^{dn}$ the following identity  holds
\be\label{penide}
\sum\limits_{g\in {\mathcal G}_{n}}~\prod\limits_{\{i,j\}\in E_g}\left[  e^{ -\b V(x_i -x_j)} -1\right]~=~
\sum_{\t\in {\mathcal T}_n}\left[\prod_{\{i,j\}\in E_\t}\left(e^{- \b V(x_i-x_j)}-1\right)\right]
e^{-\b\sum\limits_{\{i,j\}\in E_{\bm{M}(\t)}\backslash E_\t}V(x_i-x_j)}
\ee
\end{teorema}
\def\Ti{\bm{T}}
\def\Mi{\bm{M}}

\\{\bf Proof }. Let us pose shortly $V_{ij}=\b V(x_i-x_j)$. Since
${\mathcal G}_n$ is the disjoint union ${\mathcal G}_n=\biguplus_{\tau\in {\mathcal T}_n}[\tau,\bm{M}(\tau)]$ we can write
$$
\sum_{g\in {\mathcal G}_n} \prod_{\{i,j\}\in E_g}\left(e^{- V_{ij}}-1\right)~=~\sum_{\t\in {\mathcal T}_n} \sum_{g\in  [\tau,\bm{M}(\tau)]}
\prod_{\{i,j\}\in E_g}\left(e^{- V_{ij}}-1\right)~=~~~~~~~~~~~~~~~~~~~~~~~~~~~~~~~~~~~~~~~~~
$$
$$
~~~~~~~~~~~~~~~~~~~~~~=~
\sum_{\t\in {\mathcal T}_n} \prod_{\{i,j\}\in E_\t}\left(e^{- V_{ij}}-1\right)\sum_{g\in G_n\atop E_\t\subset E_g\subset E_{\bm{M}(\t)}}
 \prod_{\{i,j\}\in E_g\setminus E_\t}\left(e^{- V_{ij}}-1\right) =
$$
 $$
~~~~~~~~~~~~~~~~~~~=~ \sum_{\t\in {\mathcal T}_n} \prod_{\{i,j\}\in E_\t}\left(e^{- V_{ij}}-1\right)\sum_{E\subset E_{\bm{M}(\t)}\setminus E_\t}
 \prod_{\{i,j\}\in E}\left(e^{- V_{ij}}-1\right)~=~
 $$
 $$
~~~~~~~~~~~~~~~~\,\,\,\,\;~~~~\,=~ \sum_{\t\in {\mathcal T}_n} \left[\prod_{\{i,j\}\in E_\t}\left(e^{- V_{ij}}-1\right)\right]
 \prod_{\{i,j\}\in E_{\bm{M}(\t)}\setminus E_\t}\left[\left(e^{- V_{ij}}-1\right)+1\right]~=~
 $$
 $$
 ~~~~\;\;\;\,~~\:~~~~=
 ~ \sum_{\t\in {\mathcal T}_n} \left[\prod_{\{i,j\}\in E_\t}\left(e^{- V_{ij}}-1\right)\right]
 \prod_{\{i,j\}\in E_{\Mi(\t)}\setminus E_\t} e^{- V_{ij}}~~~~~~~~~~~\Box
$$

\\From Theorem \ref{Penid} the  corollary below easily follows.

\begin{cor}\label{cor1}  Let
$V(x)$ be a pair potential. Let $n\ge 2$.
Let  $\bm{M}: {\mathcal T}_n\to {\mathcal G}_n$ be a partition scheme in ${\mathcal G}_n$.
Then, for any  $(x_1,\dots,x_n)\in \mathbb{R}^{dn}$ the following inequality  holds
\be\label{tre}
\left|\sum\limits_{g\in G_{n}}~
\prod\limits_{\{i,j\}\in E_g}\left[  e^{ -\b V(x_i -x_j)} -1\right]\right|~
\le~
\sum_{\t\in T_n}
e^{-\b\sum\limits_{\{i,j\}\in E_{\bm{M}(\t)}\backslash E^+_\t}V(x_i-x_j)}
\prod_{\{i,j\}\in E_\t}\left(1-e^{- \b |V(x_i-x_j)|}\right)
\ee
where  $E_\tau^+$ denotes the set of edges of the tree $\tau$ with non-negative energy. That is,
\be\label{t+}
E_\tau^+=\{\{i,j\}\in E_\t:~V(x_i-x_j)\ge 0\}
\ee
\end{cor}

\\{\bf Proof}.  Let us pose once again  $V_{ij}=\b V(x_i-x_j)$. By Theorem \ref{Penid} we have that
\be\label{uno}
\left|\sum_{g\in {\mathcal G}_n} \prod_{\{i,j\}\in E_g}\left(e^{- V_{ij}}-1\right)\right|~\le ~
\sum_{\t\in {\mathcal T}_n} \prod_{\{i,j\}\in E_\t}\left|e^{- V_{ij}}-1\right|
 e^{-  \sum_{\{i,j\}\in E_{\Mi(\t)}\setminus E_\t}V_{ij}}
\ee
Using now the trick  proposed in
\cite{errata}  we observe that, for any $\t\in \mathcal T_n$
$$
 \prod_{\{i,j\} \in E_\tau} |e^{-V_{ij}}-1|=\Big[ \prod_{\{i,j\} \in E_\tau} (1-e^{- |V_{ij}|}) \Big] \
 e^{- \sum_{\{i,j\}\in E_\t\setminus E^+_\t}V_{ij}}
$$
 so that
\be\label{due}
\left[\prod_{\{i,j\}\in E_\t}\left|e^{- V_{ij}}-1\right|\right]
 e^{-  \sum_{\{i,j\}\in E_{\Mi(\t)}\setminus E_\t}V_{ij}}~=~
\left[\prod_{\{i,j\} \in E_\tau} (1-e^{-| V_{ij}|})\right]
 e^{- \sum_{\{i,j\} \in E_{\bm M(\tau)} \setminus E^+_\tau}V_{ij}}
\ee
Inserting now (\ref{due}) into (\ref{uno}) we get (\ref{tre}). $\Box$

\subsection{Partition scheme via minimum spanning tree}\label{asas}
The key point now is  to find a partition scheme $\bm M$ in such a way that is possible to find a good bound for the factor
$
\exp\{-\b\sum_{\{i,j\}\in E_{\bm{M}(\t)}\backslash E^+_\t}V(x_i-x_j)\}
$
appearing in the r.h.s. of (\ref{tre}).

\\We will construct explicitly our  partition scheme $\bm M: {\mathcal T}_n\to {\mathcal G}_n$  by first defining an auxiliary map
$\bm T :{\mathcal G}_n\to {\mathcal T}_n$ and then deriving $\bm M$ from $\bm T$ according to  the following proposition.

\begin{pro}\label{twomaps}
The following statements are equivalent.
\begin{itemize}
\item[1.]
There are  two maps $$\xymatrix{{\mathcal G}_n \ar@<.5ex>[r]^\Ti & {\mathcal T}_n \ar@<.5ex>[l]^\Mi}$$
such that $\Ti^{-1}(\tau)=\{g \in {\mathcal G}_n:\, \tau \subset g \subset \Mi(\tau)\}$ for every $\tau \in {\mathcal T}_n$.
\item[2.]
 $\Mi$ is a partition scheme in ${\mathcal G}_n$.
 \end{itemize}
\end{pro}

\\{\it Proof}. $1\Rightarrow2$.
Since $g \in \Ti^{-1}(\Ti(g))$, we have $\Ti(g) \subset g$ for all $g \in {\mathcal G}_n$. In particular,
for every tree $\tau$ we have $\Ti(\tau) \subset \tau$ which implies
$\Ti(\tau)=\tau$ because both are trees.  I.e.,  $\Ti$ is  surjective and thus the  intervals $\Ti^{-1}(\tau)$ are nonempty. This implies that
 ${\mathcal G}_n$ is the disjoint union of the intervals
 ${\mathcal G}_n = \bigcup_{\tau \in {\mathcal T}_n} \Ti^{-1}(\tau)=\bigcup_{\tau \in {\mathcal T}_n}[\t, \Mi(\t)]$.
 Hence in view
 of Definition \ref{partschem} we conclude  that  $\Mi$ is a partition scheme in ${\mathcal G}_n$.

 \\$2\Rightarrow1$. If
  $\Mi$ is a partition scheme then for any $g\in {\mathcal G}_n$ there exists a unique tree $\t\in {\mathcal T}_n$
  such that  $g\in [\t,\bm M(\t)]$. Therefore we can define the map $\bm T$ from ${\mathcal G}_n$ to ${\mathcal T}_n$
  such that, for all $g\in [\t,\bm M(\t)]$,  $\bm T(g)=\t$.
  $\Box$

\vv\vv
\\ We thus start
by first defining the map $\bm T$  from ${\mathcal G}_n$ to ${\mathcal T}_n$. In order to do that,
assume we have a function defined on the edges of the complete graph, $f: E_n \to {\mathbb R}^*$. Then, for every connected graph $g \in {\mathcal G}_n$
there is at least a tree $\tau \subset g$, among the trees $\tau' \subset g$,  which minimizes the value of $\sum_{e \in E_{\tau'}} f(e)$.
This tree is called a minimum spanning tree of $g$ w.r.t. $f$. If $f$ is such that this minimum spanning tree is unique for
each graph $g \in {\mathcal G}_n$,
then $f$ induces a map, say $\bm T_f$,  from ${\mathcal G}_n$ to ${\mathcal T}_n$ which associates to each $g \in {\mathcal G}_n$
this unique minimum spanning tree.


\\ Let now  $V$ be a stable and tempered pair potential in ${\mathbb{R}}^{dn}$ and let
$(x_1,\dots,x_n)\in \mathbb{R}^{dn}$ be given.
We would like to  consider as  the function $f: E_n \to  {\mathbb R}^*$ the one that associates to each $\{i,j\}\in E_n$ the
value $V(x_i-x_j)$. The problem is that such $f$, depending on the potential $V$ and the configuration
$(x_1,\dots,x_n)\in \mathbb{R}^{dn}$ chosen, in general does not guarantee that each $g\in {\mathcal G}_n$ contains a unique
minimum spanning tree w.r.t. to $f$.
In order to avoid multiple minima, we will  therefore modify the codomain of the function $f$ by adding $n(n-1)/2$ auxiliary ``coordinates"
(which will take only integers values), each of them one-to-one associated to an edge of  $E_n$. These coordinates
will permit to distinguish eventual multiple minimum trees.

\\To carry out this procedure, we need to
widen our framework and consider, instead of functions from $E_n$ to $\mathbb{R}^*$, a suitable class of more general
functions $f:E_n\to {\mathbb K}$ with  ${\mathbb K}$ being a totally ordered Abelian
monoid and having the key property that  $\sum_{e \in E_{\tau}} f(e)$ is different
for different trees $\tau \in {\mathcal T}_n$. We recall that a {\it totally ordered Abelian
monoid} (shortly {\it tomonid})  is a structure
 $({\mathbb K}, + ,0, \ge)$ such that $({\mathbb K}, + ,0)$
is an Abelian  (i.e. commutative) monoid, $({\mathbb K}, \ge)$ is a
totally ordered set  and for all $x, y, z\in  {\mathbb K}$
we have that $x\ge y$ implies $x + z \ge y + z$  (i.e. the total order $\ge $ is translational invariant).

\begin{defi} \label{fadm}
  Let $f: E_n \to {\mathbb K}$ where ${\mathbb K}$ is a totally ordered Abelian monoid. We say that $f$ is admissible if
  for any $\t,\t'\in {\mathcal T}_n$ such that $\t\neq \t'$ we have that $\sum_{e \in E(\tau)} f(e)\neq \sum_{e \in E(\tau')} f(e)$.
\end{defi}

\\Once an admissible function $f: E_n \to {\mathbb K}$  has been given, we can  define the
maps $\bm T_f$  from ${\mathcal G}_n$ to ${\mathcal T}_n$
$\bm M_f$  from ${\mathcal T}_n$ to ${\mathcal G}_n$ as follows.

\begin{defi}\label{hyp}
Let $\mathbb{K}$ be  a totally ordered Abelian monoid. Let $f: E_n \to {\mathbb K}$ be an admissible function.
  Then, for every $g \in {\mathcal G}_n$ there is a unique  spanning tree $\t \subset g$
  for which $\sum_{e \in E(\tau)} f(e)$ is minimum. We define the map $\bm T_f: \GG_n\to \TT_n$ such that
  $\bm T_f(g)$ is this unique minimum spanning tree of $g$.
\end{defi}


\begin{defi}\label{d4}
Let $\mathbb{K}$ be  a totally ordered Abelian monoid and let $f: E_n \to {\mathbb K}$ be an admissible  function.
We define the map $\bm M_f:{\mathcal T}_n\to{\mathcal G}_n$ such that  $\bm M_f( \tau )$ is the graph on the vertices $[n]$ whose edges are
 the $\{i,j\}$ such that $f(\{i,j\}) \geq f(e)$ for every edge $e \in E_\tau$ belonging to the path  from $i$ to $j$ through $\tau$.
\end{defi}


\\Thus we have constructed
$$\xymatrix{{\mathcal G}_n \ar@<.5ex>[r]^{\bm T_f} & {\mathcal T}_n \ar@<.5ex>[l]^{\bm M_f}}$$
Observe that $\tau \subset \bm M_f(\tau)$ and $\bm T_f(g) \subset g$.
Moreover, these two maps ${\bm T_f}$, ${\bm M_f}$ satisfy the hypothesis of the Proposition \ref{twomaps}. This is the content of the following lemma.



\begin{lem}\label{l1}
  Let $f: E_n \to  {\mathbb K}$ be an admissible function and $\tau \in {\mathcal T}_n$. Let $\bm T_f$ and $\bm M_f$ be the maps given in Definitions
\ref{hyp} and \ref{d4} respectively. Then
 $$\bm T_f^{-1}(\tau)=\{g \in {\mathcal G}_n :\, \tau \subset g \subset \bm M_f(\tau)\}$$
 and therefore $\bm M_f$  is a partition scheme in ${\mathcal G}_n$.
\end{lem}

\begin{proof}
 Let $g \in \bm T_f^{-1}(\tau)$. We have $\tau=\bm T_f(g) \subset g$. Now take $\{i,j\} \in E_g$, and let $e \in E_\tau$
 be any edge belonging to the path from $i$ to $j$ in $\tau$.
 Consider $\tau'$ the graph obtained from $\tau$ after replacing the edge $e$ by $\{i,j\}$. Clearly $\tau'$ is connected and has $n-1$ edges,
 so it is a tree. By minimality of $\tau$ we must have $f(d) \leq f(\{i,j\})$, whence $\{i,j\} \in E_{\bm M_f(\tau)}$. Therefore $g \subset \bm M_f(\tau)$.

 \\Conversely, let $\tau \subset g \subset \bm M_f(\tau)$. We must show $\bm T_f(g)=\tau$. By cardinality, it suffices to show  $\bm T_f(g) \subset \tau$.
 Proceeding by contradiction, take $\{i,j\} \in E_{\bm T_f(g)} \setminus E_\tau$. Consider the path $p^\t(\{i,j\})$ in $\tau$ joining $i$ with $j$.
   Since $\bm T_f(g) \subset \bm M_f(\tau)$, $f(\{i,j\})$ is greater or equal than the corresponding value for any
 edge in the path $p^\t(\{i,j\})$. If we remove $\{i,j\}$ from $\bm T_f(g)$, the tree splits into two trees. Necessarily, at least one of the edges in
 the path $p^\t(\{i,j\})$ joins a vertex of one tree with a vertex of the other. Thus, by adding this edge we obtain a connected graph with $n-1$ edges, a new tree,
 which contradicts the minimality of $\bm T_f (g)$.


\end{proof}
\vv

\\We now provide an explicit construction for the admissible function $f$ of Definition \ref{hyp}.

\\We start by specifying the totally ordered Abelian monoid $\mathbb{K}$.
Assume to have chosen an order in $E_n$ (e.g. the lexicographic order: $\{i,j\}<\{i',j'\}$
if either $i<i'$ or $i=i'$ and $j<j'$). Let $ \mathbb{N}_0=\mathbb{N}\cup\{0\}$
and consider the set
$$
{\mathbb N}_0^{E_n}=\overbrace{{\mathbb N}_0\times \cdots \times \mathbb N_0}^{|E_n|~{\rm  times}}
$$
such that the $m^{\rm th}$ entry (with $1\le m\le |E_n|$)
of an element $x\in {\mathbb N}_0^{E_n}$ corresponds to the $m^{\rm th}$ edge  w.r.t. the order chosen  of $E_n$. Now we set
\be\label{tomonoid}
{\mathbb{K}}\doteq {\mathbb R}^*\times {\mathbb N}_0^{E_n}
\ee
 An element of ${\mathbb{K}}$ is thus an ordered
$(|E_n|+1)$-tuple such that the first entry is a real number while the remaining $|E_n|$ entries are natural numbers or zero.

\\The set ${\mathbb{K}}$  defined in (\ref{tomonoid}) has a canonical structure of Abelian monoid and
can  also be endowed with a natural total order by
considering the lexicographical order on $\mathbb{R}^*\times \mathbb{N}_0^m$,
prioritizing the entries from left to right. It is easy to show that such an order is translational invariant respect to the
standard sum in ${\mathbb{K}}$ which is thus a totally ordered Abelian monoid.
\vv
\\We are now ready to define the function $f:E_n\to \mathbb{K}$.
Let  $\bm 1_{\{i,j\}}$ be the element of ${\mathbb N}_0^{E_n}$  with all entries zero except the one at the position
corresponding to the edge $\{i,j\}$ which is equal to one.

\begin{defi}\label{fV}
 Given a pair potential  $V$ and given $(x_1,\dots,x_n)\in  \mathbb{R}^{dn}$,  define the
function \newline
\be\label{effe}
f: ~E_n \to {\mathbb{K}}:~ \{i,j\}\mapsto V(x_i-x_j)\times \bm 1_{\{i,j\}}
\ee
\end{defi}

\\{\bf Remark}. The function  $f$  above is, for every pair potential $V$ and for any $(x_1,\dots,x_n)\in  \mathbb{R}^{dn}$, admissible
according to Definition \ref{hyp} and thus, according to Definition \ref{d4} and Lemma \ref{l1}, the map  $\bm M_f:{\mathcal T}_n\to{\mathcal G}_n$ is a
partition scheme in ${\mathcal G}_n$. This is our new partition scheme. The advantage to consider this new partition $
\bm M_f$ scheme instead of the one originally proposed
by Penrose is manifestly clear in the following key lemma which practically  concludes the proof of Proposition \ref{pro1}.



\vv


\begin{lem}\label{stabp}
 Let $V: {\mathbb R}^d \to {\mathbb R}^*$ be a stable pair potential with stability constant $B$, and $\tau \in {\mathcal T}_n$.
 Then, for every $(x_1,...,x_n) \in {\mathbb R}^{dn}$,
\be\label{Peb}
 \sum_{\{i,j\} \in \bm E_{M_f(\tau)} \setminus E_\tau^+} V(x_i-x_j) \geq -Bn
\ee
where $f$ is the function given in Definition \ref{fV}, $\bm M_f(\tau)$ is the graph given in Definition \ref{d4} and $ E_\tau^+$ is the  subset
of $E_\t$ defined in (\ref{t+}).

\end{lem}

\begin{proof}
 In the following, $V_{ij}$ denotes $V(x_i - x_j)$, and we will make implicit use of the following trivial fact:
 $$
 (x,\sigma) \geq (x',\sigma') \Rightarrow x \geq x' \ \ \mbox{ for } (x,\sigma), (x',\sigma') \in \mathbb{R}^*\times{\mathbb N}_0^{E_n}
 $$
Now we proceed to show that the inequality (\ref{Peb}) holds true.
 The set of edges $E_\t\setminus E_\tau^+$ forms  the forest $\{\tau_1,...,\tau_k\}$.
 Let us denote ${\mathcal V}_{\tau_s}$ the vertex set of the  tree $\tau_s$
 of the forest.
Assume $i \in {\mathcal V}_{\tau_a}$, $j \in {\mathcal V}_{\tau_b}$.
 If $a \neq b$, the path from $i$ to $j$ through $\tau$ involves an edge $e$ in $E_\tau^+$. Thus, if in addition $\{i.j\} \in
E_{\bm M_f(\tau)}$, we
 have $V_{ij} \geq V_e \geq 0$. If $a=b$, the path from $i$ to $j$ through $\tau$ is contained in $\tau_a$. Thus, if in addition
 $\{i,j\} \notin E_{\bm M_f(\tau)}$, we must have $V_{ij} \leq V_e \leq 0$ for some edge $e$ in that path. This allows to bound:

 $$\sum_{\{i,j\} \in E_{\bm M_f(\tau)} \setminus E_\tau^+} V_{ij} \geq \sum_{s=1}^k \sum_{\{i,j\} \subset {\mathcal V}_{\tau_s}} V_{ij} \geq
 \sum_{s=1}^k -|\mathcal V_{\tau_s}|B \ge -nB  ~~~~~~~~~~~~~~~~~~\Box$$

\end{proof}

\subsection{Conclusion of the proof of Proposition \ref{pro1}}

We now have all ingredients to conclude the proof of Proposition  \ref{pro1}. Indeed,
for any fixed configuration $(x_1,\dots,x_n)\in \mathbb{R}^{dn}$,   by Corollary   \ref{cor1} and Lemma \ref{l1} we have

 $$
 \Big| \sum_{g \in {\mathcal G}_n} \prod_{\{i,j\} \in E_g} (e^{-\b V(x_i-x_j)}-1) \Big| \leq
 \left[ \prod_{\{i,j\} \in E_\tau} (1-e^{-\b| V(x_i-x_j)|})\right]
 e^{-\b \sum_{\{i,j\} \in E_{\bm M_f(\tau)} \setminus E^+_\tau}V(x_i-x_j)}~\le $$
 $$~~~~~~~ \le e^{\b Bn}\sum_{\tau \in {\mathcal T}_n} \prod_{\{i,j\} \in \tau} (1-e^{-\b| V(x_i-x_j)|})
 $$
where $f$ is the admissible function given in (\ref{effe}),  $\bm M_f$ is the partition scheme given  in Definition \ref{d4} and
where to get  the inequality of the last line we have used Lemma
\ref{stabp}.
\section{Proof of Theorem \ref{teo1}}
To deduce Theorem \ref{teo1} from Proposition \ref{pro1}
we just need to show the following easy lemma.
\begin{lem}\label{le3}
For any  $\t\in T_n$ it holds
$$
\int_{\L}dx_1\dots \int_{\L}dx_{n} \prod_{\{i,j\}\in E_{\t}}
\left(1-e^{-\b |V(x_i-x_j)|}\right)\le
|\L|\left[\hat C(\b)\right]^{n-1}
$$
\end{lem}
\def\xx{x}
\\{\bf Proof}.
Without loss of generality  we can always assume  (eventually through a renomination of the
indices in the integral of the l.h.s. of equation  above)  that  $E_{\t}~=~\{2,j_2\},\{3,j_3\},\dots \{n,j_{n}\}$ with $j_k<k$ for all
$k=2,\dots, n$, so that
$$
\int_{\L}d\xx_1\dots \int_{\L}d\xx_{n} \prod_{\{i,j\}\in E_{\t}}
\left( 1- e^{ -\b |V(\xx_i -\xx_j)|}\right)~
=~
\int_{\L}d\xx_1\dots \int_{\L}d\xx_{n} \prod_{k=2}^{n}
\left(1-e^{ -\b |V(\xx_{k} -\xx_{j_k})|}\right )
$$

\\Define the following change of variables $(x_1,\dots,x_n) \to (y_1,\dots, y_n)$ in the integral
\be\label{change}
y_1 =x_1, ~~~~y_{k}= x_{k}-x_{j_k}, ~~~~~~\forall k =2,\dots ,n
\ee
The Jacobian matrix   of such transformation (\ref{change}), being lower triangular with entries of the diagonal all  equal to one,  has determinant
equal to one.
Therefore we have
$$
\int_{\L}d\xx_1\dots\int_{\L}d\xx_{n} \prod_{k~=~1}^{n-1} \left(1-e^{ -\b |V(\xx_{i_k} -\xx_{j_k})|}\right ) ~\le
$$
$$
\le~
\int_{\L}dy_1\int_{\mathbb{R}^3}dy_2\dots
\int_{\mathbb{R}^3}dy_{n}
 \prod_{j=2}^{n}
\left( 1- e^{ -\b |V(y_j)|}\right )~=
$$
$$
=~
|\L|\left[\int_{\mathbb{R}^3}(1-   e^{ -\b |V(x)|})dx\right]^{n-1}~= |\L|\left[\hat C(\b)\right]^{n-1}
$$
$\Box$
\vv

\\Theorem \ref{teo1} follows now straightforwardly. Indeed, from Proposition \ref{pro1} and Lemma \ref{le3} we get
 the following bound for the  absolute
value of the Mayer coefficient $C_n(\b,\L)$.
$$
|C_n(\b,\L)|~\le ~{1\over |\L|}{1\over n!}\int_{\L}\,dx_1
\dots \int_{\L} dx_n\: \left|\sum\limits_{g\in {\mathcal G}_{n}}
\prod\limits_{\{i,j\}\in E_g}\left[  e^{ -\b V(x_i-x_j)|} -1\right]\right|\le
$$
$$
~~~~~~~~~~~~~~\le ~{e^{\b Bn}\over |\L|}{1\over n!}\int_{\L}\,dx_1
\dots \int_{\L} dx_n\: \sum_{\tau \in {\mathcal T}_n} \prod_{\{i,j\} \in E_\tau} (1-e^{-\b| V(x_i-x_j)|}) =
$$
$$
~~~~~~~~~~~~~~= ~{e^{\b Bn}\over |\L|}{1\over n!}\sum_{\tau \in {\mathcal T}_n}\int_{\L}\,dx_1
\dots \int_{\L} dx_n\:  \prod_{\{i,j\} \in E_\tau} (1-e^{-\b| V(x_i-x_j)|}) \le
$$
$$
~~~~~~\le ~{e^{\b Bn}}{1\over n!}\left[\hat C(\b)\right]^{n-1} \sum_{\tau \in {\mathcal T}_n} 1~ =~{e^{\b Bn}}{n^{n-2}\over n!}\left[\hat C(\b)\right]^{n-1}
~$$
where in the last line we have also used the Cayley formula $\sum_{\tau \in {\mathcal T}_n} 1=|\mathcal T_n|=n^{n-2}$
(see \cite{Ca}).
This concludes the proof of Theorem \ref{teo1}.
\section{Comparison with recent results and concluding remarks}

As observed in Section 2, the improvement of Theorem \ref{teo1} respect to the old Penrose-Ruelle bound is
manifestly evident, as far as general stable and tempered pair potentials are concerned. There have been however recent works \cite{MPS, dLP, dLPY}
which also obtain improvements on the convergence radius of the Mayer series for systems of particles interacting
via a restricted class of stable and tempered pair potentials. In particular,  in a recent paper \cite{dLPY}, which is a development of an early work by Basuev \cite{Ba2},
the authors obtain new bounds for the Mayer series convergence radius which are better than any precedent  bound given in the
literature (see in \cite{dLPY} the new bounds of  Theorem 5, formula (4.5)), as far as  the particles in the  system interact via a so-called Basuev potential (see Definition 3 in \cite{dLPY}). This
is a class of potentials which includes the classical case of the Lennard-Jones type.

\\It is worth to say that, at least for the particular case
of the Lennard-Jones potential, the optimal bounds presented in \cite{dLPY} (obtained
also using a recent result  by one of us \cite{Y}) are slightly better than those obtained here.

\\Indeed, as far as the Lennard-Jones potential $V(r)={1\over r^{12}}-{2\over r^{6}}$ is concerned, the calculations given in Sec 5.3 of \cite{dLPY} show that, at $\b=1$
the lower bound for the convergence radius of the  Mayer series of the Lennard-Jones gas is surely greater   than
$$
 {e^{-({B_{_{\rm LJ}}}+1)}\over
7.4 }
$$
where $B_{LJ}$ is the Lennard Jones  stability constant. The same calculation  using instead the bounds of Theorem \ref{teo1} yields
for the same convergence radius a lower bound surely smaller  than
$$
 {e^{-({B_{_{\rm LJ}}}+1)}\over
8.08}
$$
The results of this paper, as well as those given in \cite{Ba2, dLP, MP, dLPY}, show that direct methods based on
tree-graph identities are, in the end, much more effective than the old indirect methods of \cite{Pe63,Ru63} based on K-S equations.
On the other hand, the fact that the bound  obtained in \cite{dLPY}  beats  our bound of Theorem \ref{teo1} for some
particular cases of stable and tempered pair potentials
raises the question of whether it is possible to further improve the bounds of Theorem \ref{teo1} in the general case.

\section*{Acknowledgments}

\\ A.P.  has been partially supported by the Brazilian  agencies CNPq
(Conselho Nacional de Desenvolvimento Cient\'{\i}fico e Tecnol\'ogico - Bolsa de Produtividade em pesquisa, grant n. 306208/2014-8)
and  FAPEMIG (Funda{\c{c}}\~ao de Amparo \`a  Pesquisa do Estado de Minas Gerais - Programa de Pesquisador Mineiro, grant n. 00230/14).
 S.Y. has been partially supported by the Argentine agency
CONICET (Consejo Nacional de Investigaciones Cient\'\i ficas y T\'ecnicas) and by
the mathematics department of UFMG.

%


\begin{thebibliography}{99}











\bibitem{Ba1}  A. G. Basuev (1978) :  {\it A theorem on minimal specific energy for classical systems}. Teoret. Mat. Fiz.
{\bf 37}, no. 1, 130--134.
\bibitem{Ba2}   A. G. Basuev (1979): {\it  Representation for the Ursell functions, and cluster estimates}. Teoret. Mat. Fiz. {\bf 39}, no. 1, 94-105.

%




\bibitem{BM}  D. Brydges, Ph. A. Martin (1999): {\it Coulomb Systems at Low Density: A Review},  J. Statist. Phys. {\bf 96}, 1163-1330.




\bibitem{Ca} A. Cayley (1889): {\it A Theorem on trees}.  Quarterly Journal of Pure and Applied Mathematics, {\bf 23}, 376--378.


\bibitem{FP}  R. Fern\'andez and A. Procacci (2007): {\it Cluster expansion for abstract polymer models.New bounds from an old approach}, Commun. Math Phys, {\bf 274}, 123--140.






\bibitem{Ga} G.  Gallavotti (1999): {\it  Statistical mechanics. A short treatise},  Springer, 1999.

\bibitem{Gr}  J. Groeneveld (1967):  {\it Estimation methods for Mayer graphical expansions}, Doctor's thesis
published in Proceedings of the Koninklijke Nederlandse Akademie vanWetenschappen,
Series 70, Nrs. 4 and 5, 451-507.

\bibitem{LP}  J. L. Lebowitz and O. Penrose (1964): {\it Convergence of Virial Expansions}, J. Math. Phys. {\bf 7}, 841-847.



%



\bibitem{JI} J. E. Jones; A. E. Ingham (1925): {\it On the calculation of certain crystal potential constants, and on the cubic
crystal of least potential energy}. Proc. Roy. Soc. Lond. A {\bf 107}, 636--653.





\bibitem{dLP}  B.N.B. de Lima and A. Procacci (2014): {\it The Mayer series of the Lennard-Jones gas: improved bounds for the convergence radius}, ,  J. Stat. Phys., {\bf 157}, n.3, 422-435.


\bibitem{dLPY}  B.N.B. de Lima,  A. Procacci and S. A. Yuhjtman (2016): {\it On stable pair potentials with an attractive tail, remarks on two papers by A. G.
Basuev}.  Comm. Math. Phys., {\bf 343}, 445--476.










%

\bibitem{MM} J. E. Mayer and M. G. Mayer (1940): {\it Statistical Mechanics},  Wiley, New York, 1940.



\bibitem{MP}  T. Morais and  A. Procacci (2014): {\it Continuous particles in the Canonical Ensemble as an abstract polymer gas}
Journal of Statistical Physics,
J. Stat. Phys., {\bf 151}, 830-845.

\bibitem{MPS} T. Morais, A. Procacci and B. Scoppola (2014): {\it On Lennard-Jones type potentials and hard-core potentials with an attractive tail}, J. Stat. Phys., {\bf 157} , p. 17-39.



\bibitem{Pe63} O. Penrose (1963): {\it Convergence of Fugacity Expansions for Fluids and Lattice Gases}, J. Math. Phys. {\bf 4},  1312 (9 pages).


 \bibitem{Pe67} O. Penrose (1967): {\it Convergence of fugacity
 expansions for classical systems}.  In {\it Statistical
 mechanics: foundations and applications}\/, A. Bak (ed.),
 Benjamin, New York.




\bibitem{errata} 
A. Procacci (2009): {\it Erratum and Addendum:``Abstract Polymer Models
with General Pair Interactions"}, J. Stat. Phys.,  {\bf 135}, 779--786.








\bibitem{Ru}  D. Ruelle (1969): {\it Statistical mechanics: Rigorous
    results}\/. W. A. Benjamin, Inc., New York-Amsterdam.

\bibitem{Ru63} D. Ruelle (1963): {\it Correlation functions of classical gases}, Ann. Phys., {\bf 5}, 109--120.


\bibitem{So1}
A. D. Sokal (2001): {\it Bounds on the complex zeros of (di)chromatic polynomials
and Potts-model partition
functions}. Combin. Probab. Comput., {\bf 10}, no. 1, 41--77.




\bibitem{Y} S. A. Yuhjtman (2015): {\it A sensible estimate for the stability constant of the Lennard-Jones potential},
J. Stat. Phys. {\bf 160}, no. 6, 1684--1695.






\end{thebibliography}
\end{document}